\documentclass[11pt]{article}

\usepackage{graphicx,psfrag}

\usepackage{amssymb,amsmath,amsthm,enumerate}
\usepackage{fullpage}
\usepackage{bbm}

\usepackage{tikz}

\usepackage{scalerel}\usepackage{setspace}

\usepackage{amsopn}

\RequirePackage[colorlinks,citecolor=blue,urlcolor=blue]{hyperref}

\newtheorem{lemma}{Lemma}
\newtheorem{remark}{Remark}

\newtheorem{definition}{Definition}
\newtheorem{theorem}{Theorem}

\newtheoremstyle{examplestyle}
  {0.5\topsep}  
  {0.5\topsep}   
  {}          
  {}          
  {\bfseries\slshape} 
  {.}         
  { }         
  {}          

\theoremstyle{examplestyle}
\newtheorem{example}{Example}

\newcommand{\cA}{\mathcal{A}}

\newcommand{\dE}{\mathbb {E}}
\newcommand{\dP}{\mathbb{P}}

\newcommand{\dZ}{\mathbb {Z}}
\newcommand{\dN}{\mathbb {N}}

\newcommand{\dR}{\mathbb {R}}
\newcommand{\dC}{\mathbb {C}}
\newcommand{\cF}{\mathcal {F}}

\newcommand{\cB}{\mathcal{B}}
\newcommand{\cG}{\mathcal {G}}

\newcommand{\cP}{\mathcal {P}}

\newcommand{\cH}{\mathcal {H}}

\newcommand{\cZ}{\mathcal {Z}}
\newcommand{\cL}{\mathcal {L}}

\newcommand{\1}{1\!\!{\sf I}}
\newcommand{\IND}{\1}

\newcommand{\cGr}{\mathcal {G}^{\scaleto{\bullet}{3pt}}}
\newcommand{\cGe}{\mathcal {G}^{\scaleto{\bullet}{3pt}\scaleto{\bullet}{3pt}}}

\newcommand{\ee}{\mathrm{e}}

\newcommand{\AUT}{{\mathrm{Aut}}}
\newcommand{\SUPP}{{\mathrm{supp}}}

\newcommand{\CAY}{{\mathrm{Cay}}}

\begin{document}

\title{Logarithmic regularity of spectral measures on infinite graphs}

\author{Charles Bordenave\thanks{CNRS, Aix-Marseille Université, I2M, Marseille, France. Email: charles.bordenave@cnrs.fr}}

\maketitle

\begin{abstract}
We study the regularity of spectral measures of self-adjoint operators on infinite weighted graphs in the unimodular setting. This framework encompasses operators in the group algebra   of a finitely generated group, random operators whose
distribution is quasi-invariant under a group action, and Benjamini--Schramm limits of operators on finite graphs. Under a natural geometric condition on the underlying graph, we prove that the expected spectral measure  satisfies a logarithmic H\"older regularity estimate. The proof relies on a strengthened version of the monotone labelling method previously introduced with Sen and Vir\'ag to control the pure point part of the spectral measure. Applications include operators in
group algebras of indicable groups, Anderson-type models with arbitrary compactly
supported potentials on Cayley graphs, anisotropic percolation operators, and operators
on quasi-transitive graphs. In particular, our results extend the classical Craig--Simon theorem beyond $\mathbb{Z}^d$.
\end{abstract}

\section{Introduction}

In this paper, we are interested in the spectrum of self-adjoint locally defined operators   on infinite graphs. The graph and the operator considered here can be deterministic or random, but they need to satisfy an invariance condition: unimodularity. This framework encompasses all operators in the group algebra $\dC[\Gamma]$ of a finitely generated group $\Gamma$. It contains more generally all random operators  whose distribution is quasi-invariant under the group action. It also includes all operators which are the Benjamini-Schramm limits of a sequence of operators on finite graphs. Under a simple geometric condition,  we notably prove that the expected spectral measure, say $\mu$, satisfies the following regularity  condition: for intervals $I$ in the spectrum, 
$$\mu(I) = O \left( \frac{ 1 }{ \ln ( C / |I|)}\right),$$ 
where $|I|$ the length of the interval. This work is to some extent a sequel of the joint paper \cite{zbMATH06821385} with Sen and Vir\'ag where under the same or similar geometric conditions, the absence/presence of atoms in the expected spectral measure was established. In the present work, we strengthen the monotone labelling tool introduced in \cite{zbMATH06821385} in two ways: we control the expected spectral measure on intervals and not only singletons; we also generalize this tool to deal with graphs with more intricate structures.

In \cite{zbMATH06821385}, the examples of applications focused  on random graphs, random trees and percolation graphs. In the present work, we explore new applications to transitive and quasi-transitive  graphs with possibly random weights. We present in this introduction three applications. Each application  is contained in the next application but to avoid introducing too many notions at the same time, we present them successively (at the expense of some redundancy).

The motivations to study the spectral measures are extremely diverse. For example, in algebraic topology, the study of the spectral measure of elements in the group algebra $\mathbb C[\Gamma]$ has attracted a lot of attention in connection with the theory of $L^2$-invariants and $L^2$-Betti numbers, we refer to the monograph \cite{zbMATH01805723}. In quantum mechanics, Wiener's Lemma implies that for any self-adjoint operator $H$ on a Hilbert space and any vector $f\in D(H)$, we have 
$$
\lim_{T \to \infty}  \frac{1}{T} \int_0^T \left| \langle f, e^{itH}  f \rangle  \right|^2 dt= \sum_{\lambda} \mu_H^f (\{\lambda\})^2,
$$
where $\mu_H^f$ is the spectral measure at vector $f$, see \eqref{eq:defmufP}, and the sum is over the set of its atoms. Hence, the diffusion of quantum particles is intimately related to the spectral decomposition. This is the content of the RAGE Theorem, we refer to \cite[Section 2]{zbMATH06532891}.

\paragraph{Operators in group algebra.}

Let $\Gamma$ be a finitely generated group with unit $\ee$ and $\lambda_{\Gamma}$ be its left-regular representation on $\ell^2 (\Gamma)$ (defined by $\lambda_{\Gamma}(g) f (x)  = f (g^{-1} x)$, for $f \in \ell^2(\Gamma)$, $g,x \in \Gamma$). An element $p = \sum_g p_g  g  \in \dC [\Gamma]$ in the group algebra defines a bounded operator on $\ell^2 (\Gamma)$ through $\lambda_{\Gamma}$ by  
\begin{equation*}
\lambda_{\Gamma} (p) = \sum_g p_g \lambda_{\Gamma}(g).
\end{equation*}
To put it plainly, for any $f \in \ell^2 (\Gamma)$ and $x \in \Gamma$ we have 
\begin{equation}\label{eq:deflambdap}
\lambda_\Gamma (p) f (x) = \sum_{g} p_g f( g^{-1} x).
\end{equation}

Since $\lambda_{\Gamma}$ is faithful, there is no harm in identifying $\lambda_{\Gamma}(p)$ and $p$. In the sequel, we often write $p$ in place of $\lambda_{\Gamma}(p)$.   If $S = S^{-1}$ is a finite symmetric set of generators and  $1_S = \sum_{g \in S} g \in \dC [\Gamma]$, then $\lambda_{\Gamma}(1_S)$ is the adjacency operator of the Cayley graph $\CAY(\Gamma,S)$ of $\Gamma$ with generators $S$ (we allow $\ee \in S$). If $\SUPP(p) \subset S$, we can interpret $p$ as a local operator on $\CAY(\Gamma,S)$. For example if $p_g \geq 0$ and $\sum_g p_g  = 1$ then $p$ is the transition kernel of a random walk on $\CAY(\Gamma,S)$.

If $p = p^*$ where $p^* = \sum_g \bar p_g g^{-1}$ then $p$ is a self-adjoint operator on $\ell^2 (\Gamma)$. From the spectral theorem, we can then decompose $p$ as
$$
p = \int \lambda  dE (\lambda),
$$
where $E$ is the resolution of the identity of $p$. The {\em spectral measure of $p$} (aka Plancherel or Kesten measure) is the probability measure on $\dR$: $\mu_p = \langle \delta_\ee ,  E \delta_\ee \rangle$, where $\delta_g$ is the Dirac delta mass at $g \in \Gamma$. Concretely, for any bounded measurable function $\varphi$, 
$$
\int \varphi (\lambda) d \mu_p (\lambda) = \langle \delta_e , \varphi (p) \delta_\ee \rangle = \int \varphi (\lambda) d \langle \delta_\ee , E(\lambda) \delta_\ee \rangle.
$$

We are interested in the spectral decomposition of $\lambda_\Gamma(p)$ when $p = p^*$. That is, we want to decompose $\mu_p$ in pure point, singular continuous and absolutely continuous parts. Beware that this decomposition may vary in nature across a given group. Indeed, for the lamplighter group $\dZ_2 \wr \dZ$, there are some generating sets $S$ and $S'$ such that $1_{S}$ has purely point spectrum while $1_{S'}$ is purely singular continuous, see \cite{zbMATH01421105,zbMATH06721412}.  If $\Gamma$ admits the free group as a finite index subgroup, then the situation is relatively well-understood: $\mu_p$ has at most finitely many atoms  \cite{Linnell1993} and its continuous part is absolutely continuous (with algebraic density), see \cite{zbMATH05996482,zbMATH07756921} and references therein. If $\Gamma$ admits  $\dZ^d$ as a finite index subgroup, the same situation prevails thanks to Bloch-Floquet theory, see notably \cite{HiguchiShirai1999,KorotyaevSaburova2017} in this direction. In torsion-free groups where the strong Atiyah conjecture holds, $\mu_p$ has no atoms if $p \ne 0$, see  \cite{zbMATH07184957,zbMATH08052380} for important recent results in this direction. Beyond these good cases, it is fair to remark that not much is known, see \cite{10.1007/3-540-16777-3_74,McLaughlin1987,https://doi.org/10.1112/blms/21.3.209,zbMATH00410028,zbMATH00950181} for early references on this topic. In this paper, we make some progress in this direction. For the related question of the connectedness of the spectrum, we refer to \cite{zbMATH00098786,zbMATH07207199} and references therein.

The following definition is central in our applications. Recall that a group $\Gamma$ is indicable if it admits a surjective homomorphism to $(\dZ,+)$.  We specialize this definition as follows.

\begin{definition}\label{def:indicable}
Let $k \geq 1$ integer, $\Gamma$ be a group,  $S = S^{-1 } \subset \Gamma$ a symmetric subset  and an element $a \in S$. We say that $(\Gamma,S)$ is {\em $(a,k)$-indicable} if there exists $\phi \in \mathrm{Hom}( \Gamma, \dZ)$ such that $k = \phi(a)> \phi(b)$ for all $b \in S$, $b \ne a$. 
\end{definition}

Note that if $(\Gamma,S)$ is $a$-indicable then $a$ has necessarily infinite order. If $k =1$, an equivalent definition is the following: $(\Gamma,S)$ is $(a,1)$-indicable if and only if for any integer $n \geq 1$ and any  $w \in S^n$, $w_1\cdots w_n = \ee$ implies that the number of $i$'s such that $w_i = a$ is equal to the number of  $i$'s such that $w_i = a^{-1}$.

As noted in \cite[Remark 1.2]{zbMATH07949207}, for any finitely generated indicable group $\Gamma$, there exists a symmetric set of generators $S$ and $a \in S$ such that $(\Gamma,S)$ is $(a,1)$-indicable. The following classical groups are  indicable: free groups, Artin groups, braid groups,  fundamental groups of Riemann surfaces or Thompson group. Equipped with their natural sets of generators $S$ and picking any $a \in S$, the following classical groups are $(a,1)$-indicable: free groups, even Artin groups or fundamental groups of Riemann surfaces.

Note however that in all these groups, it is possible to find a generating set $S$ such that $(\Gamma,S)$ is not $(a,k)$-indicable for any $a \in S$ and integer $k \geq 1$. Indeed, if $a$ and $b$ in $\Gamma$ do not commute and $S$ contains $a b^{\epsilon}$ and $b^{\epsilon}a$ with $\epsilon \in \{-1,0,1\}$ then  any  $\phi \in \mathrm{Hom}( \Gamma, \dZ)$ satisfies $\phi(a b^{\epsilon} )  = \phi( b^{\epsilon}  a) = \phi (a) + \epsilon \phi(b)$ in particular $\phi$ cannot have a unique maximum among $a b^{\epsilon}$, $b^{\epsilon}a$ and their inverses. We postpone to Subsection \ref{subsec:bonus} the application of our technique to this situation.

Our first result is a refinement of \cite{zbMATH07949207}. Here and below if $I\subset \dR$ is an interval then $|I|$ denotes its Lebesgue measure.  By convention, a singleton $I = \{\lambda\}$ is a closed interval of length $0$. For $x$ real, we set $(x)_+ = \max (x,0)$. The operator norm of an operator $T$ is denoted by $\| T \|_{\rm op}$.

\begin{theorem}\label{th:mainG}
Let $p \in \mathbb C[\Gamma]$ such that $p = p^*$ and let $S  = \SUPP(p)$. If there exists $a \in S$ such that $(\Gamma,S)$ is $(a,k)$-indicable then for any interval $I \subset \sigma(p)$ we have
$$
\mu_p(I) \leq \frac{2k\ln (\alpha)}{\left( \ln \left(  \beta / |I|  \right) \right)_+ },  
$$
where $\alpha = 3 \|p\|_{\rm op} / |p_a|$, $\beta = 2 \sqrt 2  |p_a|$. In particular, $\mu_p$ has no atom.
\end{theorem}

This theorem improves on \cite[Theorem 1.1]{zbMATH07949207} in three ways: firstly, $\Gamma$ is allowed to be non-amenable, secondly \cite[Theorem 1.1]{zbMATH07949207} considered only singleton $I = \{\lambda\}$, lastly, in our terminology, \cite{zbMATH07949207} considered only $k=1$. As mentioned above, this theorem applies for example to all adjacency operators of Cayley graphs of even  Artin groups or surface groups with their natural sets of generators. Note that for surface groups \cite{zbMATH07184957} and right-angled Artin groups \cite{zbMATH06050472}, the strong Atiyah conjecture is known to hold and  thus $\mu_p$ has no atom when $p \ne 0$. 

The conclusion of Theorem \ref{th:mainG} cannot be very significantly improved without further assumptions. Indeed, Kotowksi and Vir\`ag have proved that for $\Gamma = \dZ_2 \wr \dZ$ and some operator $p \in \mathbb C [\Gamma]$ satisfying the hypothesis of the theorem, we have 
$\mu_p( (-\epsilon,\epsilon ) ) \sim C / ( \ln \epsilon )^2$ as $\epsilon \to 0$ (see discussion in \cite{zbMATH06721412}).

When $\Gamma$ is indicable but when there does not exist $a \in S$ such that $(\Gamma,S)$ is $(a,k)$ indicable then a weaker version of Theorem \ref{th:mainG} holds. Namely,  in Subsection \ref{subsec:bonus}, we will prove Theorem \ref{th:bonus} which asserts that for any finite symmetric generating set $S$, there exists $p = p^* \in \dC [\Gamma]$ with $\SUPP(p)  = S$ such that the conclusion of Theorem \ref{th:mainG} holds. For the lamplighter group $\Gamma = \dZ_2 \wr \dZ$, this statement will have an interesting consequence, see Remark \ref{rk:GZ}.

\paragraph{Invariant random operators.}

We now generalize Theorem \ref{th:mainG} to random operators in $\cB(\ell^2( \Gamma))$ whose law is invariant by right multiplication. More precisely, fix $S = S^{-1}$ a finite symmetric set of generators (the unit $\ee$ could be in $S$). We consider an array $P = (p_g(x))_{g \in S, x \in \Gamma} \in \dC^{S \times \Gamma}$ which is {\em symmetric} in the sense that for all $x \in \Gamma$ and $g \in S$, 
\begin{equation}\label{eq:sympgx}
p_{g^{-1}} (gx) = \overline p_{g} (x).
\end{equation}
The array $P = (p_{g} (x))_{g \in S, x \in \Gamma} $ defines a symmetric operator through the formula for $f \in \ell^2 (\Gamma)$ with finite support and $x \in \Gamma$,
$$
P f (x) = \sum_{g \in S} p_g (x) f (g^{-1} x ). 
$$
In particular, from \eqref{eq:deflambdap}, if $p_{g} (x)$ does not depend on $x$, we retrieve the previous setting. Under mild assumptions, $P$ extends to a self-adjoint operator. For example if $K = \sup_{x \in \Gamma} \sum_{g \in S} |p_g(x)|< \infty$ then $P$ is a bounded operator on $\ell^2 (\Gamma)$ with $\| P \|_{\rm op} \leq K$. As above, the spectral measure of a vector $f$ is then defined by: for any bounded measurable function $\varphi$, 
\begin{equation}\label{eq:defmufP}
\int \varphi(\lambda) d \mu_P ^{f} (\lambda) = \langle f , \varphi(P) f \rangle.
\end{equation}

We are interested  in random operators, that is when the array $P = (p_{g} (x))_{g \in S, x \in \Gamma}$ is random and its distribution has the following invariance in distribution: for any $y \in \Gamma$,
$$
(p_{g}(x))_{g \in S, x \in \Gamma} \stackrel{d}{=} (p_g(x y))_{g \in S, x \in \Gamma},
$$
where $\stackrel{d}{=}$ is equality in distribution.
We then say that $P$ is a  {\em right-invariant} array. For example, if $(p_g(x))_{g \in S , x \in \Gamma}$ are independent and identically distributed variables   (up to symmetry \eqref{eq:sympgx}) then $P$ is right-invariant.

This setting allows notably to consider the Anderson tight-binding model on $\Gamma$, that is operators of the form:
$$
P = V + \lambda_{\Gamma} ( 1_{S_0} ),
$$
where $S_0 = S_0^{-1}$ is a finite symmetric set of generators of $\Gamma$ with $\ee \notin S_0$, $1_{S_0} = \sum_{g \in S_0} g$ and $V$ is the diagonal operator defined by $V f(x) = V_x f(x)$ with $(V_x)_{x \in \Gamma}$ independent and identically distributed real random variables. Indeed, set $S = S_0 \cup \{\ee \}$, $p_\ee(x) = V_x$ and $p_g (x) = 1$ for $g \in S_0$. 

This setting also contains anisotropic percolation graphs. Let $S = S^{-1}$ be a finite symmetric set of generators of $\Gamma$ and for $g \in S$, let $\pi_g  = \pi_{g^{-1}} \in [0,1]$. The percolation graph of $\CAY(\Gamma,S)$ associated to the probabilities $ (\pi_g)_{g \in S}$ is the graph whose adjacency operator is the above operator $P$ with $(p_g(x))_{g \in S , x \in \Gamma} \in \{0,1\}^{S \times \Gamma}$  independent (up to symmetry \eqref{eq:sympgx}) with $\dP ( p_g(x) = 1) = 1 - \dP ( p_g(x) = 0) = \pi_g$.

There is an extension of Theorem \ref{th:mainG} to right-invariant operators for the {\em expected spectral measure} $\dE [ \mu^{\delta_\ee}_P ]$, which is also known as the density of states in mathematical physics (here and below $\dE$ denotes the expectation with respect to the underlying randomness).

\begin{theorem}\label{th:mainGI}
Let $S = S^{-1}$ be a finite symmetric subset of $\Gamma$ and $P = (p_g(x))_{g \in S, x \in \Gamma}$ be a right-invariant symmetric random array.  Assume that there exist $a \in S$ such that $(\Gamma,S)$ is $(a,k)$-indicable and deterministic constants $K,p_* >0$ such that with probability one, $\| P \|_{\rm op} \leq K$ and for all $x \in \Gamma$, $|p_a(x)| \geq p_{*} $. Then for any interval $I \subset \sigma(P)$ we have
$$
\dE \mu^{\delta_\ee}_P (I) \leq \frac{2k \ln (\alpha)}{ \left( \ln \left(  \beta / |I|  \right) \right)_+},  
$$
where $\alpha = 3 K / p_*$ and $\beta = 2 \sqrt 2   p_*$. 
\end{theorem}

Interestingly, the claim of the theorem does not depend on the weights $p_{g}(x)$ with $g \notin \{a,a^{-1}\}$. For example, it could be applied to the Anderson tight-binding model for any potential $V_x$ whose distribution is compact. Our result can then be seen as an extension of the classical Craig-Simon theorem \cite{zbMATH03844760} (where in contrast with the classical Wegner estimates, there is no density assumption or independence that is assumed on the potential $V$, see \cite{zbMATH06532891}).  It can also be applied to anisotropic percolation with any probabilities $(\pi_g)_{g \in S}$ such that $\pi_a = \pi_{a^{-1}} = 1$. We then obtain a uniform bound (over the values of $(\pi_g)_{g \in S \backslash \{a^{-1}, a\}}$) on the logarithmic regularity. This surprising phenomenon had already been observed on $\dZ^d$ and for the absence of atoms in  \cite{zbMATH06821385} (under the name of vertical percolation).

Again, the statement of Theorem \ref{th:mainGI} is  not too far from optimal for this set of assumptions. For the Anderson tight-binding model $P$ on $\dZ$, it was known to Dyson that $\dE \mu^{\delta_\ee}_P( (-\epsilon,\epsilon)) \sim C /(\ln \epsilon)^2$ as $ \epsilon \to 0$ under mild assumption on the distribution of $V_x$, we refer again to \cite{zbMATH06721412}.

\paragraph{Quasi-transitive graphs.}

We give an illustration for (possibly random) operators on  quasi-transitive graphs.  Let $V$ be a finite set, $\Gamma$ be a finitely generated group as above and $S = S^{-1}$ be a finite symmetric subset of $\Gamma$. We now consider operators acting on $ \dC^V \otimes \ell^{2} (\Gamma)$. A vector $f \in \dC^V  \otimes \ell^{2} (\Gamma)$ can be written as $f = (f(x))_{x \in \Gamma}$ where $f(x) \in \dC^V $ is the orthogonal projection onto the vector space supported by vectors with support in $V \times \{x \}$. For each $x \in \Gamma$ and $g \in S$, let $p_g (x ) \in M_V (\dC)$ be such that 
$$
p_{g^{-1}} (g x) = p_g (x)^*,
$$
where $*$ is the conjugate transpose of a matrix. The array $P = (p_g(x))_{g \in S, x \in \Gamma} \in M_V(\dC)^{S \times \Gamma}$ defines a symmetric operator through the formula: for any $f \in \dC^V \otimes \ell^2 (\Gamma)$ with finite support and $x \in \Gamma$,
$$
(P f ) (x)  = \sum_{g \in S} p_g (x) f (g^{-1} x ),
$$
(this is an identity in $\dC^V$). Again, under mild assumptions, $P$ extends to a self-adjoint operator, for example if $\| P \|_{\rm op} \leq \sup_{x \in \Gamma} \sum_{g \in S} \| p_g (x) \|_{\rm op}$ is finite.  If $p_g(x) = p_g$ does not depend on $x$, then $P$ can simply be written as 
\begin{equation}\label{eq:defPlift}
P = \sum_{g \in S} p_g \otimes \lambda_{\Gamma} (g).
\end{equation}
This type of quasi-transitive operators  with respect to the right-action of $\Gamma$ is ubiquitous, see for example \cite{zbMATH01421105,zbMATH05177852,9317918,bordenave2025sparsegraphsbenjaminischrammlimits} and references therein. Also, elements of the group algebra of a direct or semi-direct products $G \ltimes \Gamma$ where $G$ is a finite group can be written in the form \eqref{eq:defPlift}.

More generally, when $P$ is random, we say that the array $P = (p_g(x))_{g \in S, x \in \Gamma} \in M_V(\dC)^{S \times \Gamma}$ is right-invariant, if for any $y \in \Gamma$, $P \stackrel{d}{=}(p_g(xy))_{g \in S, x \in \Gamma}$.

The last result of this introduction allows to control by induction on $S$ and $V$, the atoms of the expected spectral measure of $P$ defined as
$$
\dE [ \mu_{P}^{\delta_{\ee o}}] = \frac{1}{|V|} \sum_{v  \in V} \dE [ \mu_P^{\delta_{\ee v}} ], 
$$
where $\delta_{\ee v} = \delta_{\ee} \otimes \delta_u$ and $o$ is uniformly distributed on $V$. If $P$ has the form \eqref{eq:defPlift}, the $L^2$-Betti numbers can be expressed as the mass at $0$ of such spectral measure, see \cite{zbMATH01805723,zbMATH05177852}.

 Below, if $T \in M_V(\dC)$ and $U_1,U_2 \subset V$, $T_{|U_1 \times U_2} \in M_{U_1,U_2} (\dC)$ is the submatrix  with indices in $U_1 \times U_2$. If $U = U_1 = U_2$, we simply write $T_{|U}$ in place of $T_{|U \times U}$. Also, for $\gamma > 0$, we say that a measure $\mu$ on $\dR$ is {\em $\gamma$-log H\"older} if there are constants $c_0,c_1 >0$ such that for all intervals $I \subset \dR$, $\mu(I) \leq c_0 (\ln (c_1/|I|))_+^{-\gamma}$.

\begin{theorem}\label{th:mainGQI}
Let $V$ be a finite set, $S = S^{-1}$ be a finite symmetric subset of $\Gamma$ and  assume that there exists $a \in S$ such that $(\Gamma,S)$ is $(a,k)$-indicable for some $\phi \in \mathrm{Hom} (\Gamma,\dZ)$. Let $P = (p_g(x))_{g \in S, x \in \Gamma}$ be a right-invariant symmetric random array in  $ M_V(\dC)^{S \times \Gamma}$.  Assume that there is a partition $V = V_0 \sqcup V_1$ ($V_1$ can be empty) and deterministic constants $K,p_* >0$ such that with probability one, $\| P \|_{\rm op} \leq K$, $\| { p_a(x)_{| V_0} }^{-1} \|_{\rm op} \leq p^{-1}_{*} $ and   $p_{g}(x) _{|V  \times V_1} = 0$ for all $x \in \Gamma$ and $g \in S$ with $\phi(g) > 0$. Then for any real $\lambda$ we have
$$
\sum_{v\in V} \dE \mu^{\delta_{\ee v}}_P (\{\lambda\}) \leq  \sum_{v \in V_1}  \dE \mu^{\delta_{\ee v}}_{P_1} (\{\lambda \}),
$$
where $P_1 = (p_g (x)_{|V_1})_{g \in S_1, x \in \Gamma} \in M_{V_1}(\dC)$ and $S_1 = S \backslash \{a^{-1} , a \}$ (with the convention that the sum over an empty set is $0$). Moreover, for any $ 0 < \gamma_1 \leq 1$, if the continuous part of  $\sum_{v \in V_1}  \dE \mu^{\delta_{\ee v}}_{P_1}$ is $\gamma_1$-log Hölder,  then the continuous part of $\sum_{v\in V} \dE \mu^{\delta_{\ee v}}_P$ is $\gamma$-log Hölder for any $\gamma < \gamma_1 / (1 + \gamma_1)$. 
\end{theorem}

 This result can be applied to many quasi-transitive operators.
 Let us give for simplicity some deterministic examples with $p_g (x) = p_g = p^*_{g^{-1}}$ independent of $x$ as in \eqref{eq:defPlift}. The obvious example is when $(\Gamma,S)$ is $(a,k)$-indicable and $p_a$ is invertible: we can then directly apply the theorem to $V_0 = V$.

A simple example with $V_1$ not empty is the following: assume, in addition to the assumptions of the theorem, that $\ee \in S$ and ${p_g}_{|V_1} = 0$ for all $g \in S \backslash \{\ee\}$. Then $P_1 = p_1  \otimes 1_{\Gamma}$, where $p_1 = {p_{\ee}}_{|V_1}$ and $1_\Gamma$ is the identity operator on $\ell^2(\Gamma)$.  Let 
$$\mu_{ p_1  } = \frac{ 1}{|V_1|} \sum_{i = 1}^{|V_1|} \delta_{\lambda_i (p_1 )} $$
be the empirical distribution of the eigenvalues of $p_1$ (counting multiplicities). From the spectral theorem, we have 
$$
\mu_{p_1  }  = \frac{1}{|V_1|} \sum_{v \in V_1} \mu_{p_1}^{\delta_v}.
$$
Hence, from Theorem \ref{th:mainGQI}, for any real $\lambda$, we get
$$
\dE \mu_{P}^{\delta_{\ee o}} (\{\lambda \}) = \frac{1}{|V|} \sum_{v\in V}  \mu^{\delta_{\ee v}}_P (\{\lambda \}) \leq  \frac{|V_1|}{|V|} \mu_{p_1} (\{\lambda \}) .
$$
In particular, the set of atoms  is finite and contained in the set of eigenvalues of $p_1$. Moreover, for any $v \in V$, the continuous part of $\dE \mu_{P}^{\delta_{\ee v}}$ is $\gamma$-log Hölder for any $\gamma < 1/2$.

Our last example uses successive applications of Theorem \ref{th:mainGQI}. Assume that we have a partition $V = \sqcup_{i=1}^k V^i_0$ with $|V^i_0| = d_i \geq 1$, and assume that there exist distinct elements $\{a_i,a_i^{-1}\}$ in $S$ such that for any $i$, $(\Gamma,S)$ is $(a_i,1)$-indicable and $p_{a_i}$ is zero except on the $V^i_0 \times V^i_0$-diagonal block where it is equal to $q_i \in M_{d_i}(\dC)$ with $\| q_i^{-1} \|_{\rm op} \leq p^{-1}_*$. We can then iteratively apply Theorem \ref{th:mainGQI} for $j = 1 , \ldots , k$ to $ V^j  = \sqcup_{i \geq j} V^i_0$, $V_0^j$ and $V_1^j = \sqcup_{i > j} V^i_0$. We get for any $v \in V$ that $\mu^{\delta_{\ee v}}_P$ is continous and $\gamma$-log Hölder for any $\gamma < \psi^k(1)$, where $\psi(x) = x / (1+x)$. 

\paragraph{Proof overview. } At a high level, the existence of a surjective homomorphism $\phi$ from $\Gamma$ to $\dZ$ allows to decompose $\ell^2(\Gamma)$ along the values of $\phi(x)$, $x \in \Gamma$. We can then study the eigenvalue equation $(P - \lambda)f = 0$ iteratively along this decomposition, the idea being that  $f(x)$ is determined by the values of $f(y)$ with $\phi(y) < \phi(x)$. This strategy works nicely in finite dimension and beyond groups in the more general setting of what we called monotone labelling in \cite{zbMATH06821385}, see Lemma \ref{th:monof}. In infinite dimension, to make sense of this decomposition, we prove that such decomposition can be made inside a von Neumann algebra of finite type. The existence of a trace on this algebra will be guaranteed by the invariance properties of the operators that we consider. Then, using the von Neumann dimension associated to this trace, the argument stated in finite dimension can be made effective also in infinite dimension. This will lead to our main results on logarithmic regularity of spectral measures, Theorem \ref{th:mono} and Theorem \ref{th:monoP}.

\paragraph{Some perspectives. } Our basic approach in finite dimension, Lemma \ref{th:monof}, is rather rudimentary. It would be extremely interesting to identify deeper geometric structures to study eigenvalue equations which might take advantage of richer structures such as homomorphisms onto $\dZ^d$ or free groups. More generally, there is a lack of new methods which could establish absolutely continuous spectrum.

It would be also very natural to explore similar strategies for periodic operators on a Riemannian manifold with a properly discontinuous, cocompact and isometric $\Gamma$-action. Such operators are notably studied in  \cite{zbMATH00098786}. What appears to be a first difficulty to extend our tools in this setting is that the corresponding von Neumann algebra is not of finite type.

\paragraph{Organization of the paper. }

In Section \ref{sec:unimod}, we present the general framework of unimodular random graphs and their local operators. In Section \ref{sec:mono}, we state and prove our main result on the regularity of the spectral measure. In Section \ref{sec:group} we apply this tool and prove the theorems stated above.

\paragraph{Acknowledgments. }  

This work would never have happened without enlightening and stimulating discussions with Peter Sarnak and Cyril Letrouit. It was done during the author's stay at the Institute for Advanced Study. He is grateful to the Institute for its unique environment. The author acknowledges support from the James D. Wolfensohn Fund and grant ANR-25-CE40-5672.

\section{Unimodular graphs and local operators}\label{sec:unimod}

In this section, we introduce some well-known elements of the theory of unimodular random rooted graphs. To avoid reproducing previous texts, we have chosen to present here this theory in a restricted setting tailored for the models mentioned in the introduction. This avoids some technicalities and helps to convey some key ideas. This might also serve some pedagogical purposes. We refer to \cite{MR2354165,10.1214/17-EJP73,MR3933204} for the general theory. 

\paragraph{Marked Graph. }

A {\em graph} $G = (V,E)$ will be a pair formed by a countable vertex set $V$ and a countable edge set $E$. An edge $e = (u,v) \in E$ is an ordered pair of vertices. The graphs here are undirected meaning that $E$ is equipped with the involution $e \mapsto e^{-1}$ such that if $e = (u,v)$, $e^{-1} = (v,u)$.  For ease of notation, for $u,v \in V$, we write $u \sim v$, if  $e = (u,v) \in E$. Our graphs may have  self-loops, that is, edges such that $e = e^{-1} = (u,u)$.  
The degree $\deg(v)$ of $v \in V$ is the number of edges of the form $e = (v,u) \in E$. We say that $G$ has {\em bounded degree} if $\sup_{v \in V} \deg(v) < \infty$. 
A graph is {\em connected} if there is a path connecting any pair of vertices.

Let $\cZ$ be a set equipped with an involution denoted by $*$. A {\em marked graph} $G = (V,E,\xi)$ is a  graph $(V,E)$ and a map $\xi : E \to \cZ$ satisfying the symmetry relation: for all $e \in E$, 
\begin{equation}\label{eq:symme2} \xi(e^{-1}) = \xi (e)^*.
\end{equation}

We will call the unmarked graph $\bar G = (V,E)$ the {\em skeleton} of $G$. A {\em rooted marked graph} $g = (G,o)$ is the pair formed by a connected marked graph and a distinguished vertex $o \in V$ called the root. Similarly, an {\em edge-rooted marked graph} $(G,e)$ is a connected marked graph $G$ with a distinguished oriented pair $e = (u,v) \in V^2$ (we do not necessarily assume that $e \in E$).

\begin{example}[Theorem \ref{th:mainG}]\label{ex:1}
In the setting of Theorem \ref{th:mainG}, a Cayley graph $ G_0 = \CAY(\Gamma,S)$ of a finitely generated group $\Gamma$ with finite symmetric generated set $S$ defines a graph with vertex set $\Gamma$ and edge set $E  =\{ (x,gx) : x \in \Gamma, g \in S \} \simeq \Gamma \times S$. If $p \in \dC[\Gamma]$ with $p = p^*$ and $\SUPP(p) \subseteq S$, then $G_1 = (G_0,\xi)$ defines a marked graph on the mark space $\cZ = \dC$ equipped with the involution $z^* = \bar z$ where the edge $(x,gx)$ receives the mark $\xi(x,gx) = p_{g^{-1}}$.
\end{example}

\begin{example}[Theorem \ref{th:mainGI}]\label{ex:2}
 In the setting of Theorem \ref{th:mainGI}, the symmetric array $P = (p_g(x))_{g \in S,x \in \Gamma} \in \dC^{S \times \Gamma}$ defines a marked graph $G_2$ on the mark space $\cZ =  \dC$. The edge $ (x , gx)$ receives the mark $p_{g^{-1}}(x)$. 
\end{example}

\begin{example}[Theorem \ref{th:mainGQI}] \label{ex:3} In the setting of Theorem \ref{th:mainGQI}, the symmetric array $P = (p_g(x))_{g \in S,x \in \Gamma} \in M_V(\dC)^{S \times \Gamma}$ defines a marked graph $G_3$ on the vertex set $V \times \Gamma$ where the edge $((x,u),(gx,v))$ has mark $p_{g^{-1}}(x) (u,v) \in \dC$.
\end{example}

\paragraph{Weighted adjacency operator. } In the next section, we will study (weighted) adjacency operators.  A marked graph $G = (V,E,\xi)$ on the  mark space $\cZ = \dC$ (with  $z^* = \bar z$) will be called a {\em weighted graph}. We define its adjacency operator by the formula, for all compactly supported $f \in \ell^2(V)$,
\begin{equation}\label{eq:defA}
A_G f (u) = \sum_{v \sim u} \xi(u,v) f(v).
\end{equation}
From \eqref{eq:symme2}, $A_G$ is symmetric. Hence if $\| A_G \|_{\rm op} < \infty$, it defines a bounded self-adjoint operator on $\ell^2(V)$. We denote the spectral measure of $A_G$ at vector $f \in \ell^2(V)$ as $\mu_G^f$, that is for any bounded measurable function $\varphi$, 
\begin{equation}\label{eq:defmuG}
\int \varphi(\lambda) d\mu_G^f( \lambda ) = \langle f , \varphi( A_G ) f \rangle .
\end{equation}

For example, the marked graphs $G_1$, $G_2$ and $G_3$ of Examples \ref{ex:1}, \ref{ex:2} and \ref{ex:3} are weighted graphs. The operators $\lambda_{\Gamma}(p)$ and $P$ mentioned in Theorem \ref{th:mainG}, \ref{th:mainGI} and \ref{th:mainGQI} are the adjacency operators $A_{G_1}$, $A_{G_2}$ and $A_{G_3}$. We explain below that the adjacency operators belongs to a unital $*$-algebra of operators. But, before that, we  first  define the local topology and unimodularity.

\paragraph{Local topology. } 
For simplicity of exposition, we fix some connected bounded degree graph $G_0 = (V,E)$  and a complete separable mark space $\cZ$. We let $\cG = \cG(G_0,\cZ)$ be the set of marked graphs of $G = (V,E,\xi)$  with skeleton $\bar G = G_0 = (V,E)$. Similarly, we let $\cGr = \cGr(G_0,\cZ)$ and  $\cGe = \cGe(G_0,\cZ)$ denote the set of rooted (respectively edge-rooted) marked graphs $(G,o)$ and $(G,e)$ such that $G \in \cG$, $o \in V$ and $e \in V^2$ ($e$ is not necessarily an edge of $G_0$).

Fix an increasing connected exhausting sequence $(V_k)_{k \in \dN}$ of $V$. If $ (G,o)$ is a rooted marked graph, we  define $(G,o)_k$ as the rooted marked graph spanned by vertices in $V_k$ (with the convention that $(G,o)_k$ is an empty graph if $o \notin V_k$). This family of maps $(G,o) \mapsto (G,o)_k$, $k \in \dN$, defines a projective system. The local topology on $\cGr$ is the product topology inherited from these projections.  With this topology, $\cGr$ is a complete separable metric space. 
  The local topologies on $\cG$ and $\cGe$ are defined analogously. 

We denote by $\cP(\cGr)$ the set of probability measure on $\cGr$ equipped with the weak topology. A measure $\rho \in \cP(\cGr)$ defines a random rooted marked graph $(G,o)$. The expectation of an integrable function $f$ on $\cGr$ with respect to $\rho$ will be denoted by 
$$\dE_\rho f (G,o) = \int_{\cGr} f (G,o) d \rho (G,o).$$

\paragraph{Unimodularity.}
We fix some subgroup $\AUT_0$ of the automorphism group of $G_0$. We say that a measurable function $f : \cGe \to \dR$ is {\em invariant} if for any $\phi \in \AUT_0$ and $(G,u,v) \in \cGe$, we have $f(\phi(G),\phi(u),\phi(v)) = f(G,u,v)$. A probability measure $\rho \in \cP(\cGr)$ is {\em unimodular} if for every non-negative invariant function $f$ on $\cGe$, we have  
\begin{equation}
\label{eq:defunimod} 
\dE_{\rho} \sum_{v \in V} f(G,o,v) = \dE_{\rho} \sum_{v \in V} f(G,v,o).
\end{equation}

By extension, we say that a random rooted marked graph $(G,o)$  is unimodular if its distribution is unimodular.  The basic example is the following. Let $G \in \cG$, $\AUT_0$ be the trivial subgroup (only the identity) and assume that $V$ is finite. If $o$ is uniformly distributed on $V$ then $(G,o)$ is unimodular. Indeed for any function $f \in \cGe$,
 \begin{eqnarray*}
|V| \cdot  \dE  \sum_{v \in V } f (G,o,v)  =  \dE  \hspace{-3pt} \sum_{ (u,v) \in  V^2 } f (G,u,v)   =  |V| \cdot  \dE  \sum_{u \in V } f (G,u,o) .
 \end{eqnarray*}
This observation is particularly important as it implies that  Benjamini-Schramm limits of finite graphs are unimodular, for precisions see \cite{MR2354165,zbMATH06902684,bordenave2025sparsegraphsbenjaminischrammlimits}.

\setcounter{example}{0}

\begin{example}[Continued]  We continue with the notation of Example \ref{ex:1}.
In the setting of Theorem \ref{th:mainG}, the marked graph $ G_1$ defines a  unimodular measure $\rho = \delta_{(G_1,\ee)}$.  Indeed, let $\AUT_0$ be the right multiplication by elements of $\Gamma$ and $G_0 =  \CAY(\Gamma,S)$. Since $G_1$ is invariant by right multiplication, for any invariant function $f$ and $x,y,g \in \Gamma$, we have $f (G_1,x,y) = f (G_1,xg,yg) $. We find:
$$
\sum_{x \in \Gamma} f (G_1,\ee,x) = \sum_{x \in \Gamma} f (G_1,x^{-1},\ee) = \sum_{x \in \Gamma} f (G_1,x,\ee).
$$ 
The above computation can also be directly obtained from the unimodularity (in the measured group theory sense) of the counting measure $\sum_{x \in \Gamma} \delta_{x}$, see \cite{MR2354165}.
\end{example}

\begin{example}[Continued]
Let $\AUT_0$ and $G_0  = \CAY(\Gamma,S)$ be as in the previous example. In the setting of Theorem \ref{th:mainGI}, the right-invariant symmetric array $P = (p_g(x))_{g \in S,x \in \Gamma} \in \dC^{S \times \Gamma}$ defines a unimodular rooted marked graph $(G_2,\ee)$  on the vertex set $\Gamma$. Indeed, for any invariant function $f$ and $g \in \Gamma$, we have $\dE f (G_2,x,y) = \dE f (G_2,xg,yg) $. From Fubini's Theorem, we get as above
$$
\dE \sum_{x \in \Gamma} f (G_2,\ee,x) = \sum_{x \in \Gamma} \dE f (G_2,x^{-1},\ee) = \dE \sum_{x \in \Gamma} f (G_2,x,\ee).
$$ 
\end{example}

\begin{example}[Continued] In the setting of Theorem \ref{th:mainGQI}, let $G_0$ be the tensor product of the complete graph on $V$ (loops included) and $\CAY(\Gamma,S)$ of Example \ref{ex:1}.  Let $\AUT_0$ be the right action of $\Gamma $ on $\Gamma\times V$ defined for $g \in \Gamma$ by $(x,v) \mapsto (xg,v)$. The right-invariant symmetric array $P = (p_g(x))_{g \in S,x \in \Gamma} \in M_V(\dC)^{S \times \Gamma}$ defines a unimodular rooted marked graph $(G_3 , (\ee,o) )$ on the vertex set $\Gamma \times V$ where $o \in V$ is uniformly distributed on $V$. Indeed, by right-invariance of $P$, for any invariant function $f$ and $g \in \Gamma$, we have $\dE f (G_3,(x,v),(y,v)) = \dE f (G_3,(xg,v),(yg,v)) $.  We get
 \begin{eqnarray*}
|V| \cdot  \dE \hspace{-3pt}\sum_{(x,v) \in \Gamma \times V } f (G_3,(\ee,o),(x,v)) & = & \dE  \hspace{-3pt} \sum_{(x,u,v) \in \Gamma \times V^2 } f (G_3,(\ee,u),(x,v)) \\
& = &  \sum_{(x,u,v) \in \Gamma \times V^2 }  \dE f (G_3,(x^{-1},u),(\ee,v)) \\
& = & |V| \cdot  \dE \hspace{-3pt}\sum_{(x,u) \in \Gamma \times V } f (G_3,((x,u),(\ee,o)) .
 \end{eqnarray*}
\end{example}

\paragraph{Algebra of local operators. }
We finally define a unital $*$-algebra of operators associated to (random) marked graphs. Let $G = (V,E,\xi) \in \cG$ and let  $ p : \cGe \to \dC$ be a measurable function.  We  define an operator on compactly supported functions $f \in \ell^2 (V)$ by, for all  $u \in V$,
\begin{equation}\label{eq:defpG}
p_G f (u) = \sum_{v \in V} p(G,u,v) f(v).
\end{equation}
For a weighted graph, the adjacency operator $A_G$ in \eqref{eq:defA} is an example of such operator with $p(G,u,v) = \xi(u,v) \IND( u \sim v)$.
Let $\cA_G$ be the subset of invariant functions $ p : \cGe \to \dC$ such that  for some $K \geq 1$, for all $u,v \in V$, $ | p(G,u,v)| \leq K $ and $p(G,u,v) = 0$ if $d(u,v) \geq K$ (the graph distance in $G_0$ between $u$ and $v$). 
Setting $p^* (G,u,v) = \overline{p(G,v,u)}$, the set $\cA_G$ defines a unital $*$-algebra of bounded operators on $\ell^2(V)$ (since $G_0$ has bounded degree).

\setcounter{example}{0}

\begin{example}[Continued]  We continue with the notation of Example \ref{ex:1}. If $G_0 = \CAY(\Gamma,S)$ with $S = S^{-1}$ finite generating $\Gamma$, the $*$-algebra $\cA_{G_0}$ reduces to the group algebra $\dC [\Gamma]$.  
\end{example}

As explained in  \cite{MR2354165,zbMATH05700352,zbMATH06902684}, we can extend this to random operators.  Let $\rho \in \cP(\cGr)$ be a probability on random rooted marked graphs $(G,o)$. Let $\cA_\rho$ be the subset of invariant functions $ p : \cGe \to \dC$ such that $\rho$-a.s.\  for some $K \geq 1$, for all $u,v \in V$, $ | p(G,u,v)| \leq K $ and $p(G,u,v) = 0$ if $d(u,v) \geq K$ (functions are defined up to $\rho$-null sets). Then $\cA_\rho$ defines a unital $*$-algebra of bounded operators on the Hilbert space $\cH_\rho = \int^{\oplus} \ell^2(V) d \rho(G,o)$ (a direct integral). If $\rho$ is unimodular, then the algebra $\cA_\rho$ has a state: 
\begin{equation}
\tau (p) = \dE_{\rho} [ \langle \delta_o, p_G \delta_o \rangle ] = \dE_{\rho} [ p(G,o,o) ].
\end{equation}
This state is positive and faithful. Indeed, we have
$$
\tau ( p p ^*) = \dE_{\rho} \sum_{v \in V} |p(G,o,v)|^2 \geq 0 
$$
with equality if and only if $\rho$-a.s.\ the event $E_o$ occurs, where for $u \in V$, $E_u = \{\hbox{$\forall v \in V$ : } p(G,u,v) = 0 \}$.  However we apply \eqref{eq:defunimod}  to $f(G,u,v) = \IND ( E_u ) \IND ( d(u,v) \leq k) $ which is invariant if $p$ is invariant. We get $ \dE \sum_{v \in B(o,k)} \IND ( E_v) = 0$ where $B(o,k)$ is the ball of radius $k$. Since $k$ is arbitrary and $G_0$ is connected, we get $\rho$-a.s.\ $p(G,u,v) = 0$ for all $u,v \in V$ as requested.

This state is also tracial: for $p,q \in \cGe$,
$$
\tau( p q) = \dE_{\rho}  \sum_{v \in V} p (G,o,v) q (G,v,o) =   \dE_{\rho}  \sum_{v \in V} p (G,v,o) q (G,o,v) = \tau (qp),
$$
where we have applied \eqref{eq:defunimod} to $f(G,u,v) = p (G,u,v) q (G,v,u) $.

We denote by $\cL_\rho$ the von Neumann algebra associated to $\cA_\rho$ (its weak-$*$ closure). This von Neumann algebra $\cL_\rho$ equipped with its faithful, normal tracial state $\tau$ will play an important role in  the sequel.

We shall notably use its {\em von Neumann dimension}. We say that a closed vector subspace $H$ of $\cH_\rho$ is {\em invariant} if, $P_H$, the orthogonal projection onto $H$, is in $\cL_\rho$. Then, the dimension of $H$ is 
\begin{equation}\label{eq:defdim}
\dim ( H ) = \tau (P_H) = \dE_{\rho} [ \| P_H \delta_o \|^2] \in [0,1],
\end{equation}
where $\| \cdot \|$ is the Euclidean norm in $\ell^2(V)$. If $H_1, H_2$ are closed invariant subspaces then $H_1 \cap H_2$ and $H_1 + H_2$ are also closed and invariant and the dimension formula holds:
\begin{equation}\label{eq:dimfor}
\dim ( H_1 + H_2 ) = \dim (H_1) + \dim (H_2) - \dim (H_1 \cap H_2),
\end{equation}
see \cite[exercice 8.7.31]{zbMATH00913566} and \cite[Theorem 1.12(2)]{zbMATH01805723}.

\section{Monotone labelling revisited}\label{sec:mono}

In this section, we give with Theorem \ref{th:mono} a quantitative improvement over the monotone labelling technique introduced in \cite{zbMATH06821385}. We also prove an extension to partition labelling,  Theorem \ref{th:monoP}, which will be used in the proof of Theorem \ref{th:mainGQI}.

\paragraph{Monotone labelling.}

We borrow from \cite{zbMATH06821385} the following central definition. 

\begin{definition}\label{def:label}
Let $G = (V,E)$ be a graph. For a given map $\eta : V \to \dZ$, we partition the vertices $V$ into three disjoint subsets: a vertex $x \in V$ is 
\begin{enumerate}[(i)]
\item  {\em Prodigy}, if there exists $\hat x \sim x $ such that $\eta( \hat x) < \eta(x)$ and for all $y \sim \hat x$, $y \ne x$, $\eta(y ) < \eta(x)$; 
\item {\em Level}, if $x$ is not a prodigy and for all $y \sim x$, $\eta(y) \leq \eta(x)$; 
\item {\em Bad} otherwise.
\end{enumerate}
We say that it is {\em l-bad} if it is level or bad. 
\end{definition}

\paragraph{Finite graphs.}
As a warmup, we start with a finite weighted graph $G = (V,E,p)$.  We denote by $(\lambda_i)_{1 \leq i \leq |V|}$ the eigenvalues counting multiplicities of $A_G$ defined in  \eqref{eq:defA}. The spectral counting measure of $A_G$ is defined as 
\begin{equation}\label{eq:defLG}
L_{G} =  \sum_{k=1}^{|V|} \delta_{\lambda_k}.
\end{equation}
This is a measure on $\dR$ with total mass $|V|$ and whose mass on an interval $I$ is the number of eigenvalues in $I$. We say that two intervals $I,J$ are equivalent, denoted by $I \sim J$, if $I$ and $J$ have the same center. We have the following extension of \cite[Theorem 2.2]{zbMATH06821385}.

\begin{lemma}\label{th:monof}
Let $G = (V,E,p)$ be a finite weighted graph and $\eta$ a labelling with $n$ distinct values. Assume $p_{*} = \inf_{x} |p(x,\hat x)| >0$ where the infimum is over all $x$ prodigy and $\hat x$ is  as  in Definition \ref{def:label}(i).  Let $\alpha =   3 \| A_G\|_{\rm op} / p_*$ and $\kappa = 3  \|A_{G}\|_{\rm op} \sqrt{ \alpha^2 -1} \geq 6 \sqrt{2} \|A_{G}\|_{\rm op}$. If $LB$ is the set of l-bad vertices, for any interval $I$ of length at most $ \kappa \alpha^{-n}$, we have  
$$
L_{G} ( I ) \leq  |LB|.
$$

Moreover, there exists $\beta = \beta ( \|A_G\|_{\rm op}, p_*) \geq 1$ such that for any $\theta >1$ and any interval $I$ of length at most $\beta^{-n}$, we have  
$$
L_{G} ( I ) \leq  |B|  + \sum_j L_{G_j} ( I_n ),
$$
where $B$ is the set of bad vertices, $G_j$ is the restriction of $G$ to level vertices with label $j$ and $I_n \sim I$ is the closed interval of length $|I|^{1/ (6 n  \ln n)}$. 
\end{lemma}

For adjacency operators, Lemma \ref{th:monof} applied to $I = \{ \lambda \}$ is exactly \cite[Theorem 2.2]{zbMATH06821385}. The constant $3$ in the definition of $\alpha$ and $\kappa$ is somewhat arbitrary, the proof gives a constant $2.09 \pm 0.01$. An explicit value of $\beta$ could be extracted from the proof. 

\begin{proof}[Proof of Lemma \ref{th:monof}]
Since the spectrum is contained in the closed interval $I_0 = [-\|A_G\|_{\rm op} ,\|A_G\|_{\rm op} ]$, without loss of generality, we can assume that $I \subset I_0$.  In the sequel, $\lambda \in I$ is the center of the interval $I$.  Without loss of generality, we may assume that the range of $\eta $ is $\{0,1, \ldots, n-1\}$. 

We start with the proof of the first statement which is simpler but conveys the main idea. 
Let $LB$ be the set of l-bad vertices and $P_j$ be the sets of  prodigy vertices of label $j$. We denote by $\cL\cB$ and $\cP_j$ the vector subspaces of $\dC^V$ spanned by vectors supported on $LB$  and $P_j$. Let $\cF$ be the vector space spanned by eigenvectors of $A_G$ with eigenvalues in $I$. We have $\dim(\cF) = L_G(I)$ and $\dim (\cL\cB) = |LB|$.

Let $\cF' = \cF \cap \cL\cB^\perp$ where $\perp$ denotes the orthocomplement. From Since $\dim ( \cF + \cL\cB^\perp) = \dim ( \cF  ) + \dim(\cL\cB^\perp) - \dim (\cF')$, we deduce from the codimension inequality 
$$
\dim(\cF) \leq \dim( \cF') + \dim(\cL\cB),
$$  

Therefore, the proof of the lemma will be complete if we prove that $\cF'$ is trivial.  To this end, since any vector of $\cF'$ is $0$ on l-bad vertices, we may write $\cF' = \bigoplus_j \cF'_j $ with $\cF'_j = \cF' \cap \cP_j$. We take a vector $f \in \cF'$ and set $f_j \in \cF'_j$ to be its projection onto $\cP_j$.

If $j=0$ then $\cP_0$ is trivial (no prodigy can have minimal label $0$). In particular $f_0 = 0$. We prove by induction on $j = 1, \ldots,n-1$ that 
\begin{equation}\label{eq:HR}
\| f_j \| \leq  \frac{|I|}{p_*}  \| f \|  \alpha^{j-1} .   
\end{equation}

We denote by $T_j$ the operator from $\cP_j$ to $\dC^V$ defined for $x \in P_j$ by $T_j \delta_x = p(\hat x, x) \delta_{\hat x}$ where $\hat x \sim x$ as in the definition of prodigy. The definition implies that for any $x \ne y \in P_j$, $\hat x \ne \hat y$, hence $T_j$ restricted to its image, say $\hat \cP_j$, is invertible and 
$$
\| T_j ^{-1} \|_{\rm op}  = \frac{1}{\min_{x \in P_j} |p(x,\hat x) |}\geq \frac{1}{p_*},
$$
If $f_{\leq j} = f_1 + \ldots + f_j$, we get for any $x \in P_{j+1}$,  
$$
(A_G-\lambda) f (\hat x) =  p(\hat x, x) f(x)  + (A_G-\lambda) f_{\leq j} (\hat x)
$$
or, equivalently if $\Pi_{j+1}$ is the projection onto the image of $T_{j+1}$,  
\begin{equation}\label{eq:Tj}
\Pi_{j+1} (A_G-\lambda) f  =  T_{j+1} f_{j+1}  + \Pi_{j+1} (A_G-\lambda) f_{\leq j}.
\end{equation}

Since $f \in \cF$ and $\lambda$ is the center of $I$, we have $\| (A_G-\lambda) f \|\leq (|I|/2) \| f \|$. Also, $\| A_G-\lambda \|_{\rm op} \leq 2 \| A_G\|_{\rm op}$, where we have used that $|\lambda| \leq \| A_G\|_{\rm op} $. We get that 
\begin{eqnarray*}
\| f_{j+1} \| & \leq & \| T^{-1}_{j+1} \|_{\rm op} \left (\| (A_G-\lambda) f \| + \| (A_G-\lambda) f_{\leq j} \| \right)  \\
& \leq & \frac{|I|  \| f \|}{p_*}  \left( \frac 1 2   + \IND_{j \geq 1} \frac{2 \| A_G \|_{\rm op} }{p_*} \frac{ \alpha^{j}}{ \sqrt{\alpha^2 -1}}\right),
\end{eqnarray*} 
where we have used Pythagoras Theorem and the induction hypothesis \eqref{eq:HR} to get
\begin{equation}\label{eq:pyth}
\| f_{\leq j} \|^2 \leq \sum_{k=1}^{j} \| f_k \|^2 \leq \left( \frac{\| f \| |I|}{ p_* } \right)^2 \sum_{k=1}^j \alpha^{2(k-1)} \leq  \left( \frac{\| f \| |I|}{ p_* } \right)^2 \frac{\alpha^{2j} }{\alpha^2 - 1}.
\end{equation}
For $j = 0$, we obtain the claimed estimate \eqref{eq:HR}. For $j \geq 1$, since $\| A_G \|_{\rm op} /p_* = \alpha /3$ and $\alpha \geq 3$, we arrive at for $j \geq 1$,
$$
\| f_{j+1} \| \leq \frac{|I|  \| f \|}{p_*} \alpha^{j} \left( \frac{1}{2\alpha^j} + \frac{2 \alpha}{  3 \sqrt{\alpha^2 -1}}  \right) \leq  \frac{|I|  \| f \|}{p_*} \alpha^{j} \left( \frac{1}{6} + \frac{2}{ \sqrt{3^2 -1}}  \right) \leq \frac{|I|  \| f \|}{p_*} \alpha^{j}. 
$$
This proves \eqref{eq:HR} for all $j=1,\ldots,n-1$. In particular, from \eqref{eq:pyth} applied to  $j = n-1$ and $f_{\leq n-1} = f$, we find
$$
\| f \|  \leq    \frac{\| f \| |I|}{ p_* }  \frac{\alpha^{n-1} }{\sqrt{\alpha^2 - 1}} =   \| f \|  \frac{ |I| \alpha^{n} }{\kappa}.
$$
Hence, if $|I| < \kappa \alpha^{-n}$ then $f = 0$. This proves that $\cF'$ is empty and concludes the proof of the first statement.

The second statement is a refinement which takes into account level vertices. Up to increasing $\beta$, we may assume $|I| \leq 1$ and $n \geq n_0$. Let $\theta > 1$, $s_j = c_\theta /(j+1)^\theta$, $c_\theta= (\theta -1)/(2\theta)$, and set $S_j = \sum_{i \leq j } s_i$.  We have $S_j \leq 1/2$. For some $C_0 \geq 2$ to be defined later, let $\alpha_j = C_0^j |I|^{1-S_j} $ and $I_j \sim I$ be the closed interval of length $|I|^{s_j}$.

Let $B$, $P_j$ and $L_j$ be the set of bad, prodigy and level vertices of label $j$. Let $\cF$ be as above and $\cF_j$ be the vector space of eigenvectors of $G_j$ in $I_j$. Extending a vector in $\cF_j$ by zeros, $\cF_j$ is a closed invariant subspace of $\dC^V$ supported on $L_j$. As above, let $\cB$ and $\cP_j$ be the vector space of vectors supported on $B$ and $P_j$ respectively.  We set $\cF' = \cF \cap \cB^\perp \cap \cF_j^\perp$, from \eqref{eq:dimfor}, we find
 as above that $$
\dim(\cF) \leq \dim(\cF') + \dim (\cB) + \sum_j \dim (\cF_j).
$$
By construction, we have $\dim(\cF)  = L_G(I) $,  $\dim (\cB) = |B|$ and $ \dim (\cF_j) = L_{G_j} (I_j) \leq L_{G_j} (I_n)$. It is thus sufficient to check that $\cF' = 0$ when  $2 \alpha_{n-1} \leq 1$ (we may then choose any $\beta \geq C_0^{2}$ and $\theta = 1 + 1/\ln n$ so that $s_{n-1} \geq 1/ (6 n \ln n)$ for $n$ large enough).

To this end, let $f \in \cF'$,   $f_j$, $f_{\leq j}$ be its restriction to vertices of label $j$, up to $j$, $f^+_j$ its restriction to $P_j$ and $f^o_j$ its restriction to $L_j$.   The outer boundary of $L_j$, $\partial L_j$, is the set of $x \in V$ such that $x \notin L_j$ but $x \sim y$ for some $y \in L_j$. Observe that $\partial L_j$  is included in the union of $P_j$ and vertices with label less than $j$. Hence 
\begin{equation}\label{eq:Lj}
\partial L_j \subset B \cup P_j \cup_{i < j} (L_i \cup P_i).
\end{equation}

We prove by induction on $j$, 
\begin{equation}\label{eq:HR}
\|f_j \| \leq \alpha_j \| f\|.
\end{equation} 

Note that $f^+_0 = 0$ since no vertex of label $0$ is a prodigy. Also since $f = 0$ on $B$, we deduce from \eqref{eq:Lj} that $f$ is $0$ on $\partial L_0$. In particular, if $\Pi_j^o$ is the orthogonal projection on $L_j$, we find
$$
\Pi^o_0 (A_G - \lambda) f = \Pi^o_0 (A_G - \lambda) f^o_0 = (A_{G_0} - \lambda) f^o_0.
$$ 
We have $\| \Pi^o_0 (A_G - \lambda) f \| \leq (|I|/2) \| f \|$ and $\| (A_{G_0} - \lambda) f^o_0 \| \geq (|I|^{S_0}/2) \| f^o_0\|$ since $f \in \cF$ and $f^o_0 \in \cF_0$ respectively.  We get $ \| f^o_0\| = \| f_0\| \leq |I|^{1-s_0} \| f \| = \alpha_0 \| f \|$. It proves  \eqref{eq:HR} for $j=0$.

We are ready for the induction. Assume that \eqref{eq:HR} for all $i \leq j$. In particular, if $C_0 \geq \sqrt 2$,  we have 
\begin{equation}\label{eq:fleqj}
\| f_{\leq j} \|  = \sqrt{ \sum_{i=0}^j \| f_i\|^2 } \leq \sqrt 2 |I|^{1-S_j} C_0^j \| f \| \leq \sqrt 2 \alpha_j \| f \|,
\end{equation}
 (recall that $\sum_{0 \leq i \leq j} x^i \leq 2 x^j$ if $x \geq 2$).
 Let $T_j$ be the operator from $\cP_j$ to $\hat \cP_j$ defined above. Equation \eqref{eq:Tj} gives
$$
\Pi_{j+1} (A_G-\lambda) f  =  T_{j+1} f^+_{j+1}  + \Pi_{j+1} (A_G-\lambda) f_{\leq j},
$$
where $\Pi_{j+1}$ is the orthogonal projection onto $\hat \cP_{j+1}$. Using again that $\| T_{j+1}^{-1} \|_{\rm op} \leq p_*^ {-1}$ and $\| A_G-\lambda\|_{\rm op} \leq 2 \| A_G \|_{\rm op}$, we get, since $S_{j} \geq 1/2$, for some $C > 2$,
$$
\| f^+_{j+1} \| \leq  \frac{1}{p_*} \left( |I| +  2 \sqrt 2 \| A_G \|_{\rm op}   \alpha_j  \right) \| f \| \leq C \alpha_j \| f \|_{\rm op}.
$$
Similarly, from \eqref{eq:Lj},
$$
\Pi^o_{j+1} ( A_G - \lambda ) f  = \Pi^o_{j+1} ( A_G - \lambda ) f^+_{j+1} +  \Pi^o_{j+1}  ( A_G - \lambda ) f_{\leq j } + \Pi^o_{j+1}  ( A_{G} - \lambda ) f^o_{j+1} .
$$
By construction, $ \Pi^o_{j+1}  ( A_{G} - \lambda ) f^o_{j+1}  =  ( A_{G_{j+1}} - \lambda ) f^o_{j+1} $. Hence, since $f^o_{j+1} \in \cF_{j+1}$, we find, for some new $C > 1$, 
$$
|I|^{s_{j+1}} \| f^o_{j+1} \| \leq |I| \| f \| + 2 \|A_G\|_{\rm op} \| f^+_{j+1}\| + 2 \|A_G\|_{\rm op}  \| f_{\leq j }\| \leq C \alpha_j \| f \|.
$$
Therefore, 
$$
\| f^o_{j+1} \| \leq C \alpha_j |I|^{-s_{j+1}} \| f \|.  
$$
So finally, for some new $C  > 2$, 
$$\| f_{j+1} \| \leq \sqrt{\| f^o_{j+1} \|^2 + \| f^+_{j+1} \|^2 }  \leq   C \alpha_j |I|^{-s_{j+1}}  \| f \|. $$
We fix $C_0$ to be equal to this $C$ and this proves the induction hypothesis. For $j = n-1$, we find from \eqref{eq:fleqj},
$
\| f \| =  \| f_{\leq n-1} \| \leq \sqrt 2 \alpha_{n-1} \| f \|.
$
The conclusion follows.
\end{proof}

\paragraph*{Unimodular graphs.}

We extend Theorem \ref{th:monof} to infinite unimodular graphs. Let $(G,o)$ with $G = (V,E,p)$ be a unimodular weighted graph as defined in Section \ref{sec:unimod}. We assume without loss of generality that the skeleton graph $G_0$ has loops at all vertices.  Let $\eta$ be a measurable labelling of the vertices of $G$ (possibly on an enlarged probability space). We enlarge the mark space $\cZ = \dC \times \dZ$ and form a $\cZ$-marked graph $(G,\eta)$ by adding the labels of all vertices. We then say that $\eta$ is an {\em invariant labelling}, if the rooted marked graph $(G,\eta,o)$ is unimodular.

Let $A_G$ be the weighted adjacency operator of $G$ defined in \eqref{eq:defA}. Recall the definition for $x \in V$ of $\mu_G^{\delta_x}$ in \eqref{eq:defmuG}.

\begin{theorem}\label{th:mono}
Let $(G,o)$ be a unimodular weighted  rooted  graph and $\eta$ an invariant labelling with $n$ distinct values. Assume that there exist deterministic $K , p_*  >0$ such that a.s.\ $\| A_G\|_{\rm op} \leq K$ and $\inf_{x} |p(x,\hat x)| \geq p_* $ where the infimum is over all $x$ prodigy and $\hat x$ is  as  in Definition \ref{def:label}(i). Let $\alpha =   3 K  / p_*$ and $\kappa = 3  K \sqrt{ \alpha^2 -1} \geq 6 \sqrt{2} K$. If LB is set of l-bad vertices, for any closed interval $I$ of length at most $ \kappa \alpha^{-n}$, we have  
$$
\dE [ \mu^{\delta_o}_{G} ( I ) ] \leq \dP ( o \in LB) .
$$
Moreover, there exists $\beta = \beta ( K, p_*) \geq 1$ such that for any  interval $I$ of length at most $\beta^{-n}$, we have  
$$
L_{G} ( I ) \leq  |B|  + \sum_j L_{G_j} ( I_n ),
$$
where $B$ is the set of bad vertices, $G_j$ is the restriction of $G$ to level vertices with label $j$ and $I_n \sim I$ is the closed interval of length  $|I|^{1/ (6 n  \ln n)}$. 
\end{theorem}

If $|V| < \infty$ then Theorem \ref{th:mono} reduces to Lemma \ref{th:monof}. Indeed, if $o$ is uniformly distributed on $V$, the spectral theorem implies that $|V |\cdot \dE [ \mu_{G}^{\delta_o} ] = \sum_{v \in V} \mu_{G}^{\delta_v}= L_G$ with $L_G$ as in \eqref{eq:defLG}. We note also that for adjacency operators, if $I = \{ \lambda \}$ then Theorem \ref{th:mono} is \cite[Theorem 2.3]{zbMATH06821385}. 
With the improvement of Lemma \ref{th:monof}, the proof is an easy extension of the proof of \cite[Theorem 2.3]{zbMATH06821385}.

\begin{proof}[Proof of Theorem \ref{th:mono}] Let $\rho$ be the law of $(G,\eta,o)$, let $\cA_{\rho}$ be the unital $*$-algebra of local operators and let $\cL_\rho$ be its associated von Neumann algebra of bounded operators on $\cH_\rho = \int^{\oplus} \ell^2(V) d\rho(G,o)$ defined in Section \ref{sec:unimod}. Recall that a vector  subspace $H$ of $\cH_\rho$ is invariant if the orthogonal projection onto $H$ is in $\cL_\rho$.

We essentially repeat the proof of Lemma \ref{th:monof}. Without loss of generality, we assume that the range of $\eta $ is $\{0,1, \ldots, n-1\}$, that $I \subset I_0 = [-K ,K]$ and let $\lambda \in I$ be the center of the interval $I$. 

For the first statement, let $LB$ be the set of bad vertices and $P_j$ be the sets of  prodigy vertices of label $j$. We denote by $\cL\cB$  and $\cP_j$ the vector subspaces spanned by vectors supported on $LB$ and $P_j$. All these vectors spaces are closed and invariant since $\eta$ is invariant. Let $\cF $ be the vector space spanned by the image of $\IND_I (A_\rho)$. Since $I$ is closed and $A_G$ is an element of $\cA_\rho$, $\cF$ is closed and invariant.

We denote by $\dim$ the von Neumann dimension defined in \eqref{eq:defdim}. By definition, we have $\dim(\cF) = \dE \mu^{\delta_o}_G(I)$ and $\dim (\cL\cB) = \dP ( o \in LB)$. Let $\cF' = \cF \cap \cL\cB^\perp $ where $\perp$ denotes the orthocomplement. Since $\cF$ is contained in the sum of $\cF'$, $\cL\cB$ and the $\cL_j$'s, we deduce from \eqref{eq:dimfor} that 
$$
\dim(\cF) \leq \dim( \cF') + \dim(\cL\cB).
$$ 
At this stage, we can reproduce the proof of Lemma \ref{th:monof}. The same minor modification applies to the proof of the second statement. \end{proof}

\begin{remark}
As explained in the preamble of Section \ref{sec:mono}, we have chosen to define unimodular random marked graphs in the restricted setting of a fixed skeleton graph for simplicity. It should however be noted that the proof of Theorem \ref{th:mono} holds in full generality (as it was done in \cite[Theorem 2.3]{zbMATH06821385}).
\end{remark}

\paragraph*{Monotone block labelling.}

 We now present an extension of Theorem \ref{th:mono} for labelling of blocks of vertices instead of individual vertices. This extension will be used in the proof of Theorem \ref{th:mainGQI} and  in the forthcoming Subsection \ref{subsec:bonus}. We first extend Definition \ref{def:label} to vertex partitions:

\begin{definition}\label{def:labelP}
Let $\Pi$ be a countable set, $G = (V,E)$ be a graph and $ \{ V_b \}_{b \in \Pi}$ a partition of $V$ into blocks indexed by $\Pi$. For $b,c \in \Pi$, we write $b \sim c$ if there exist $x \in V_b$ and $y \in V_c$ such that $x \sim y$. For a given map $\eta : \Pi \to \dZ$, we partition the blocks into three disjoint subsets: a block $b \in \Pi$ is 
\begin{enumerate}[(i)]
\item  {\em Prodigy}, if there exists $\hat b \sim b $ such that $\eta( \hat b) < \eta(b)$ and for all $c \sim \hat b$, $c \ne b$, $\eta(c ) < \eta(b)$; 
\item {\em Level}, if $b$ is not a prodigy and for all $c \sim b$, $\eta(c) \leq \eta(b)$; 
\item {\em Bad} otherwise.
\end{enumerate}
We say that it is l-bad if it is level or bad.
We say that a vertex $x \in V$ is prodigy/level/bad/l-bad if it belongs to a block with this property. For $x \in V_b$, let $\hat \eta (x) = (b,\eta(b)) \in \Pi \times \dZ$. If $(G,o)$ is a unimodular marked graph, we say that $\hat \eta$ is an invariant block labelling, if the marked graph $(G,\hat \eta,o)$ is unimodular. 
\end{definition}

Note that the blocks of the partition are not necessarily finite. As above, let $(G,o)$ with $G = (V,E,p)$ be a unimodular weighted  rooted  graph as defined in Section \ref{sec:unimod}. 
We have the following extension of Theorem \ref{th:mono} to invariant block labelling.

\begin{theorem}\label{th:monoP}
Let $(G,o)$ a unimodular weighted  rooted  graph and $\hat \eta$ an invariant block labelling with $n$ distinct values and with partition $\{V_b\}_{b \in \Pi}$. Let $K,p_*$ be positive real numbers. Assume that  a.s.\ $\| A_G\|_{\rm op} \leq K$. For all prodigy $b \in \Pi$, let $p_b = (p(G,x,y))_{x \in V_{\hat b} , y \in V_b}$  with $\hat b$ as  in Definition \ref{def:labelP}(i) and assume that it defines a.s.\ an invertible map $\ell^2(V_b) \to \ell^2(V_{\hat b})$ whose inverse has operator norm at most $ p_*^{-1}$. Let $\alpha =   3 K  / p_*$ and $\kappa = 3  K \sqrt{ \alpha^2 -1} \geq 6 \sqrt{2} K$. If LB is set of l-bad vertices, for any closed interval $I$ of length at most $ \kappa \alpha^{-n}$, we have  
$$
\dE [ \mu^{\delta_o}_{G} ( I ) ] \leq \dP ( o \in LB) .
$$
Moreover, there exists $\beta = \beta ( K, p_*) \geq 1$ such that for any interval $I$ of length at most $\beta^{-n}$, we have  
$$
L_{G} ( I ) \leq  |B|  + \sum_j L_{G_j} ( I_n ),
$$
where $B$ is the set of bad vertices, $G_j$ is the restriction of $G$ to level vertices with label $j$ and $I_n \sim I$ is the closed interval of length $|I|^{1/ (6 n  \ln n)}$. 
 \end{theorem}

\begin{proof}
The proof is a repetition of the proofs of Lemma \ref{th:monof} and Theorem \ref{th:mono}. For the second statement, the closed vector spaces $\cL\cB$, $\cP_j$, $\cF$ and $\cF'$ are as in the proofs of Lemma \ref{th:monof} and Theorem \ref{th:mono}. The goal is to prove that $\cF'$ is trivial. The only difference is that the operator $T_j$ on $\cP_j$ defined below \eqref{eq:HR} is now given by the formula, for $ x \in V_b$ and $b$ prodigy with label $j$,  $T_j \delta_x = p_b \delta_x \in \ell^2(V_{\hat b}) $. Let $\hat \cP_j$ be the image of the operator $T_j$. By Definition \ref{def:labelP}, if $b \ne c$ are prodigy then $\hat b \ne \hat c$. Hence $\hat \cP_j = \oplus_{b \in P_j}\ell^2 (V_{\hat b})$  and $T_j$ decomposes diagonally on this direct sum. We deduce by assumption that $T_j$ restricted to $\cP_j$ satisfies a.s.\ $\| T_j^{-1} \|_{\rm op} \geq 1/ \min_{b \in P_j} s_b \geq 1/ p_*$ with $s_b$ the smallest singular value of $p_b$ on $\ell^2 (V_b)$. The rest of the proof is unchanged.  

For the first statement, the only minor modification is that the outer boundary of $L_j$, $\partial L_j$, is defined as the set of vertices in a block say $c$, such that $V_c \cap L_j = \emptyset$ but $c \sim b$ with $V_b \subset L_j$. Observe that $\partial L_j$  is included in the union of $P_j$ and blocks with label less than $j$. Hence \eqref{eq:Lj} also holds for blocks. The rest of the proof is unchanged. \end{proof}

\section{Some applications to transitive and quasi-transitive graphs}
\label{sec:group}

In this section we prove the results stated in introduction. Since Theorem \ref{th:mainGQI} implies Theorem \ref{th:mainGI} which implies Theorem \ref{th:mainG}, it is sufficient to prove Theorem \ref{th:mainGQI}. However, for the clarity of the exposition, we start with a proof of Theorem \ref{th:mainG}.

\subsection{Proof of Theorem \ref{th:mainG}}
Let $G_0 = \CAY(\Gamma,S)$ and $G = (G_0,p)$ the weighted graph associated to $p \in \dC[\Gamma]$ as in Example \ref{ex:1}, with $G = G_1$. We  give a probabilistic proof. With $k$ as in the statement of the theorem, for an integer $n \geq 2k$, let $\omega \in \{0,1,\ldots,n-1\}$ be a uniform random variable. 

Now, let $\phi \in \mathrm{hom}(\Gamma,\dZ)$ be as in Definition \ref{def:indicable} and set, for $x \in \Gamma$, $\eta (x) = \phi(x) + \omega \; \mathrm{mod}(n)$. This defines a random labelling of $\Gamma$ with $n$ distinct values. This labelling is invariant. Indeed, for any $x,y \in \Gamma$, we have $\eta(xy) = \phi(x) + \phi(y) + \omega \; \mathrm{mod}(n)$.  However, for any integer $j$, $ j + \omega \; \mathrm{mod}(n)$ has the same distribution as $\omega$. In particular, for any $y \in \Gamma$, the array $(\eta(xy))_{ x \in \Gamma}$ is equal in distribution to $(\eta(x))_{ x \in \Gamma}$, the invariance of $\eta$ follows.

Next, observe  from the definition that $\phi(a^{-1}) = -k$ and $|\phi(b)| < k$ for all $b \in S \backslash\{a^{-1},a\}$. Hence, if $x \in \Gamma$ has label $\eta(x) = j \geq 2k$, then it has a neighbor $\hat x = a^{-1} x$ with label $\eta(\hat x) = j-k$ and the other neighbors $y$ of $\hat x$ have labels $j-2k$ (for $y = a^{-2}x$) or in $\{j-2k+1,\ldots,j-1\}$ (for $y = b a^{-1} x$, $b \in S \backslash \{a , a^{-1}\}$). Hence $x$ is a prodigy. If $\eta (x) \in \{0,\ldots,2k-1\}$ then $x$ could be l-bad and this occurs with probability $2k/n$ since $\omega$ is uniform. From Theorem \ref{th:mono}, we get that if $I \subset \sigma(p)$ is a closed interval of length at most  $ \kappa \alpha^{-n}$, we have  
$$
\mu_p (I) = \mu^{\delta_{\ee}}_{G_1} ( I ) \leq \dP ( \ee \hbox{ is l-bad}) \leq \frac{2k}{n},
$$
(this bound is also trivially true for $n < 2k$). We take $$n = \lfloor \frac{\ln (\kappa / |I| ) }{ \ln  \alpha  }\rfloor \geq \frac{ \ln (\kappa / |I| ) }{ \ln \alpha} - 1 = \frac{\ln (\kappa / (\alpha |I|) ) }{\ln \alpha}.$$
Since $\kappa / \alpha \geq (6 \sqrt{2}  K ) / ( 3K / p_*) = \beta$, the conclusion follows. \qed

\subsection{Proof of Theorem \ref{th:mainGQI}}

Let $\phi \in \mathrm{hom}(\Gamma,\dZ)$ be as in Definition \ref{def:indicable}. Recall $G_0$ and $G = G_3$ be as in Example \ref{ex:3}. Since we have assumed that a.s.\ $p_{g^{-1}} (x)_{| V_1 \times V} =  p_{g} (g x)_{| V \times V_1}^* = 0$ for all $g\in S$ such that $\phi(g) > 0$, we can remove from $G_0$ all edges of the form $e = ( (gx,u), (x ,v) )$ with $x \in \Gamma$, $u \in V$ and $v \in V_1$ and their inverses $e^{-1 } = ( (x ,v) , (gx,u) )$. We consider the deterministic partition of $V \times \Gamma$ indexed  by $\cB = \Gamma \times \{0,1\}$, $\{ V_{x,\epsilon} \}_{(x,\epsilon) \in \cB}$ given by $V_{x,\epsilon} = \{ (x,v) : v \in V_{\epsilon}\} = V_\epsilon \times \{x \}$.

The proof is very similar to the proof of Theorem \ref{th:mainG}. For integer $n \geq 2k$,  set for $(x,\epsilon) \in \cB$, $\eta (x,\epsilon) = \phi(x) + \omega \; \mathrm{mod}(n)$ where $\omega$ is uniformly random and independent of $P = (p_g(x))_{g \in S,x \in \Gamma} \in M_V^{S \times \Gamma}(\dC)$. This defines a random block labelling with $n$ distinct values. Arguing as in the proof of Theorem \ref{th:mainG}, this block labeling is invariant.

By definition $\phi(a) = - \phi(a^{-1}) = k $ and $|\phi(g)| < k $ for all $g \in S \backslash\{a,a^{-1}\}$. If the block $b = (x,0) \in \Gamma$ has label $\eta(b) = j \geq 2k$, then it has a neighbor $\hat b = (a^{-1} x,0)$ with label $\eta(\hat b) = j-k$ and the other block neighbors $c$ of $\hat b$ have labels $j-2k$ (for $c = (a^{-2}x,\epsilon)$), $j-k$ (for $c = ( a^{-1} x,1)$) or at most $j-1$ (for $c = (g a^{-1} x,\epsilon)$, $g \notin \{a , a^{-1}\}$). Hence the block $b$ is a prodigy. On the other hand, if  $b = (x,1) \in \Gamma$ has label $\eta(b) = j \geq k$, then it is level. Indeed, since there is no edge $((gx,\epsilon),b)$ with $\phi(g) > 0$, the set of neighbors of $b$ is in  $\{ ( x,\epsilon), (g x , \epsilon) : \epsilon \in \{0,1\},  g \in S : \phi(g) \leq 0\}$ whose labels are in $\{j-k,\ldots,j\}$.  If $\eta (b) \leq 2k$ then $b$ can be bad and this occurs with probability $2k/n$ since $\omega$ is uniform.

We notice finally that for any $j \geq 1$, if $G_j$ is the restriction of $G$ to level $j$ vertices and the block $(\ee,1)$ is level $j$, then the connected component of $G_j$ containing $(\ee,1)$ has the same distribution for each $j$. Namely, its skeleton $\bar G_1$ is the tensor product between the complete graph on $V_1$ (with loops) and $\CAY ( \Gamma_1, S_1)$ where $S_1 = S \backslash \{a,a^{-1}\}$ and $\Gamma_1$ is the subgroup generated by $S_1$. The weights on this skeleton graph $\bar G_1$ are given by the right-invariant array $P_1 = (p_g(x)_{|V_1})_{g \in S_1, x \in \Gamma_1}$.  

The conclusion follows from an application of Theorem \ref{th:monoP}. More precisely, we fix $\theta > 1$ and consider an interval $I$ of length at most $1$. We take $n = \lfloor (\ln (1/|I|)/ \ln \beta )^{\delta} \rfloor$ with $\delta = \gamma_1 / (1 + \theta \gamma_1)$ and $\beta$ as in  Theorem \ref{th:monoP}. Since $\ln (1/|I|^{c/n^\theta})  = O ( \ln (1/|I|) )^{1-\theta\delta }$ and $\gamma_1 ( 1- \theta \delta) = 1/ ( 1+ \theta \gamma_1)$ , we obtain that the claim of Theorem \ref{th:mainGQI} by taking arbitrarily close to $1$. \qed

\subsection{Indicable groups}
\label{subsec:bonus}

In this subsection, we discuss the situation when $\Gamma$ is  indicable but  when the assumptions of Theorem \ref{th:mainG} are not fulfilled for $S = \SUPP(p)$, $p = p^* \in \dC[\Gamma]$. We notably have the following:

\begin{theorem}\label{th:bonus}
Let $\Gamma$ be an indicable group and $S = S^{-1}$ be a finite symmetric generating set of $\Gamma$. There exist $p  = p^* \in \dC [\Gamma]$ with support $S$ and constants $c_1,c_2 >0$ such that for any interval $I \subset \sigma(p)$, we have 
$$
\mu_p( I) \leq \frac{c_1} {\ln ( c_2 / |I|) }.
$$ 
In particular, $\mu_p$ has no atom.
\end{theorem}

\begin{proof}
Fix a surjective  $\phi \in \mathrm{Hom}( \Gamma, \dZ)$. Let $S^+_0 \subset S$ be the set of $a \in S$ such that $\phi(a) = \max_{s \in S } \phi(s)$. Note that $ k = \max_{s \in S} \phi(s) \geq 1$ since $S$ generates $\Gamma$. We set $S^+_0 = \{a_1,\ldots,a_d\}$ for some arbitrary ordering. 
We claim that if $p = p^* \in \dC[\Gamma]$ has support $S$ and
\begin{equation}\label{eq:gerch}
|p_{a_1}| > \sum_{i=2}^d |p_{a_i}|
\end{equation}
then $\mu_p$ has logarithmic regularity (where the sum over an empty set is zero by convention). If $d=1$, then the statement is already proved in Theorem \ref{th:mainG}.

We assume $d \geq 2$. Let $n \geq 2k$ be an integer, $\omega \in  \Pi = \{0,1,\ldots,n-1\}$ be uniformly distributed and $\eta (x) = \phi(x) + \omega \; \mathrm{mod}(n)$ as in the proof of Theorem \ref{th:mainG}. For $b \in \Pi$, let $V_{b} \subset \Gamma$ be the subset of elements of $\Gamma$ such that   $\eta(x) = b$. This defines an invariant block labelling in the sense of Definition \ref{def:labelP} where $b \in \Pi$ has label $b$.  If $b \geq 2k$, setting $\hat b = b - k$, we see that $b$ is prodigy and that the probability that $\ee$ is l-bad is at most $2k/n$. Moreover, the restriction  $p_b$ of $\lambda(p)$ from $\ell^2 (V_b)$ to $\ell^2 (V_{\hat b})$ is the restriction from $\ell^2 (V_b)$ to $\ell^2 (V_{\hat b})$ of 
$$
\lambda(q) = \sum_{i=1}^d p_{a_i^{-1}} \lambda(a^{-1}_i) = \sum_{i=1}^d \bar p_{a_i} \lambda(a^{-1}_i).
$$

Observe that $\lambda(a^{-1}_i)$ defines an isometry from $\ell^2 (V_b)$ to $\ell^2(V_{\hat b})$. In particular, from the triangle inequality, for any $f \in \ell^2 (V_b)$, $\| p_b f \| = \| \lambda(q) f \| \geq (|p_{a_1}| - \sum_{i=2}^d |p_{a_i}| ) \| f \|$ and, under the condition \eqref{eq:gerch}, $p_b$ is invertible with bounded inverse. The claim follows from an application of Theorem \ref{th:monoP}. \end{proof}

\begin{remark}\label{rk:GZ}
On the lamplighter group $\Gamma = \dZ_2 \wr \dZ$, Grigorchuk and {\.Z}uk \cite{zbMATH01421105} give a generated set $S = (g_1,g_2,g_1^{-1},g_2^{-1})$  such that for the obvious surjective $\phi \in \mathrm{Hom}(\Gamma,\dZ)$, $\phi(g_1) = \phi(g_2) = 1$ and  $\mu_p$ is purely atomic when $p  = \IND_{S} = g_1 + g_2 + g_1^{-1} + g_2^{-2}$. On the other end, the proof of Theorem \ref{th:bonus} implies that  $\mu_p$ is purely continuous for any $p = p_1 g_1 + p_2 g_2 + \bar p_1 g^{-1} + \bar p_2 g_2^{-2}$ with $|p_1| \ne |p_2|$. From this perspective, the celebrated example of Grigorchuk and {\.Z}uk depicts  a rather exceptional behaviour. 
\end{remark}

\bibliographystyle{abbrv}
\bibliography{bib}
\end{document}